\documentclass[journal, onecolumn ,12pt]{IEEEtran}
 
\usepackage{blindtext}
\usepackage{relsize}
\usepackage{amsmath}
\usepackage{amssymb}
\usepackage{amsthm}
\usepackage{amsfonts}
\usepackage[nospace,noadjust]{cite}
\newtheorem*{theorem}{Theorem}

\usepackage{cite}

\usepackage{url}

\usepackage{textcomp}
\usepackage{bbm}
\usepackage[ruled, linesnumbered]{algorithm2e}

\usepackage{setspace}
\usepackage{graphicx}

\usepackage[caption=false,font=footnotesize]{subfig}

\usepackage{fixltx2e}

\begin{document}

\title{Flexible Allocation of Heterogeneous Resources to Services on an IoT Device}

\author{

\IEEEauthorblockN{Vangelis Angelakis, Ioannis Avgouleas, Nikolaos Pappas and Di Yuan}

\IEEEauthorblockA{Department of Science and Technology,
 Link{\"o}ping University, Campus Norrk{\"o}ping, 60 174, Sweden\\
Emails: vangelis.angelakis@liu.se, ioannis.avgouleas@liu.se, nikolaos.pappas@liu.se, diyua@itn.liu.se}

}

\maketitle
\IEEEpeerreviewmaketitle
\begin{abstract}
In the Internet of Things (IoT), devices and gateways
may be equipped with multiple, heterogeneous network interfaces which should be utilized by a large number of services. In this work, we model the problem of assigning services' resource demands to a device's heterogeneous interfaces and give a Mixed Integer Linear Program (MILP) formulation for it. For meaningful instance sizes the MILP model gives optimal solutions to the presented computationally-hard problem. We provide insightful results discussing the properties of the derived solutions with respect to the splitting of services to different interfaces.
\footnote{The research leading to these results has received funding from the European Union's Seventh Framework Programme (FP7/2007-2013) under grant agreement no [609094] (RERUM) and also under the REA grant agreement no [612361] (SOrBet). Preliminary results of this work have been presented at \cite{INFOCOM_POSTER_2015}.}

\end{abstract}

\section{Introduction}
Within the Internet of Things, resource-constrained devices may be called to provide an unpredictable set of services. Recent architectural frameworks (see e.g. \cite{RERUM01,RERUM02} and the references therein) call for de-verticalization of solutions, with applications being developed, independently of the end devices which may be anything from a sensor to the latest smartphone. This heterogeneity of devices and resources they provide to developing IoT applications is at the core of our work.

We focus on IoT networking devices having multiple, different interfaces, each of which has access to a collection of finite heterogeneous resources such as downlink data rate, buffer space, CPU interrupts, and so forth. We also consider that each service is characterized by a set of demands that can be served by the resources available on one device's interfaces. Assuming a middleware has already assigned a service onto a given device, in this work we address the problem of flexibly mapping the service resource demands onto the interfaces of that device. The flexibility of the services lies on the assumption that a demand may be served by more than one of the available interfaces, in case the available resource does not suffice, or the cost of utilizing resources over different physical interfaces proves beneficial. At the end of the day, the derived mapping can be viewed as a new virtual interface, with a one-to-one mapping of services to such dedicated virtual interfaces.
 
In general, such multi-resource allocation problems cannot be turned into single-resource ones by interchanging different resources: clearly a demand for downlink data rate cannot be exchanged with transmit buffer space. However, current literature mostly addresses such multi-resource problems, as such as scheduling jobs, often as a single-resource problem (e.g., the Hadoop and Dryad schedulers). Interface virtualization standards on the other hand deal only with same type of interfaces. Our work here presents a mixed-integer linear programming formulation of the problem of assigning services to heterogeneous interfaces with different resources. Although, the problem is computationally hard, for reasonable instance sizes, with respect to number of interfaces, types of resources, and services, optimal solutions can be derived. Initial results outline the role of different costs in the resulting flexibility of the service splitting over different interfaces.

\section{System Model}

We consider that we have a set $\mathcal{I}$ of $i$ interfaces. The interfaces are characterized by a set $\mathcal{K}$ of $k$ resources associated with them (for example {CPU cycles, Downlink capacity, Buffer size}). We assume that each service $j \in \mathcal{J}$ is associated a $K$-dimensioned demand integer vector $\mathbf{d_j}$. Likewise each interface has a $K$-dimensioned resources availability integer vector $\mathbf{b_i}$. We consider the case in which services are flexible and can utilize resources of different interfaces, with appropriate costs to model the job management overheads that will be imposed upon the operating system of the device carrying the interfaces.  We finally make the assumption that the given assignment is a feasible one, i.e. $ \sum_{j \in \mathcal{J}}d_{jk} \leq \sum_{i \in \mathcal{I}} b_{ik} $. Our goal is to assign all services to the physical interfaces, minimizing the cost of using the interfaces. We call this the \textit{Service-to-Interface Assignment (SIA) }problem. 

In the model that follows we use variable $x_{ijk}$  for the amount of the $k$-th resource of the $i$-th interface utilized by job $j$. We consider these values to be integer like the ones in the demands vectors. We denote $c_{ik}$ the per-unit cost to utilize resource $k$ on interface $i$, while with $F_{i}$ the activation cost of interface $i$. 
We also assume that each job $j$ incurs an overhead on the resource it utilizes, which may vary by interfaces in order to capture MAC and PHY layer realities, this is denoted $a_{ijk}$. 
Thus, our model amounts to:

\begin{equation}
	\begin{aligned}
 & \underset{}{\text{min.}}
 & & \sum\limits_{\mathnormal{k} \in \mathcal{K}} \sum\limits_{\mathnormal{i} \in \mathcal{I}} \mathnormal{c}_{ik} \sum\limits_{\mathnormal{j} \in \mathcal{J}} \mathnormal{x}_{ijk} + \sum\limits_{\mathnormal{i} \in \mathcal{I}} \sum\limits_{\mathnormal{j} \in \mathcal{J}} \mathnormal{F}_{i}\mathnormal{ACT}_{ij},
   \end{aligned} \label{eq:minOptFormulation}
\end{equation} 
\begin{equation}
 \begin{aligned}
 & \text{s.t.}
 & & \sum\limits_{\mathnormal{i} \in \mathcal{I}} x_{ijk} = d_{jk}, ~\forall \mathnormal{j} \in \mathcal{J}, ~\forall \mathnormal{k} \in \mathcal{K},
  \end{aligned} \label{eq:demandsConstraints}
\end{equation}
\begin{equation}
	\begin{aligned}
& & & \sum\limits_{\mathnormal{j} \in \mathcal{J}} (1+a_{ijk})x_{ijk} \leq b_{ik}, ~\forall \mathnormal{i} \in \mathcal{I}, ~\forall \mathnormal{k} \in \mathcal{K},
	\end{aligned} \label{eq:capacityConstraints}
\end{equation}
\begin{equation}
\begin{aligned}
 & & & x_{ijk} \geq 0, ~\forall \mathnormal{i} \in \mathcal{I}, ~\forall \mathnormal{j} \in \mathcal{J}, ~\forall \mathnormal{k} \in \mathcal{K},
\end{aligned} \label{eq:nonNegativityConstraints}
\end{equation}
\begin{equation}
\begin{aligned}
  & & & ACT_{ij} = \mathbbm{1} \big(  \sum\limits_{\mathnormal{k} \in \mathcal{K}} \mathbbm{1} \left( x_{ijk} > 0 \big) \right), ~\forall \mathnormal{i} \in \mathcal{I}, ~\forall \mathnormal{j} \in \mathcal{J}. 
\end{aligned} \label{eq:ACT} 
\end{equation}

Where the objective of (\ref{eq:minOptFormulation}) is to minimize the total cost of two terms: the first aims to capture the total cost incurred by the utilization of the resources over heterogeneous interfaces, the second term captures the cost introduced by splitting the service over multiple interfaces, since with each additional interface utilized the overall cost is encumbered by another $F$-term. The set of constraints in (\ref{eq:demandsConstraints}) ensures that all services demands are met, while the constraints of (\ref{eq:capacityConstraints}) ensure that the service allocation will be  performed on interfaces with available resources. In (\ref{eq:ACT}) the $\mathbbm{1}(.)$ symbol denotes the indication function becoming one if the argument is positive, zero otherwise, thus, $ACT_{ij}$ is one if and only if there is at least one resource utilizing interface $i$ for service $j$.

\begin{figure}
\centering
\includegraphics[width=1\columnwidth]{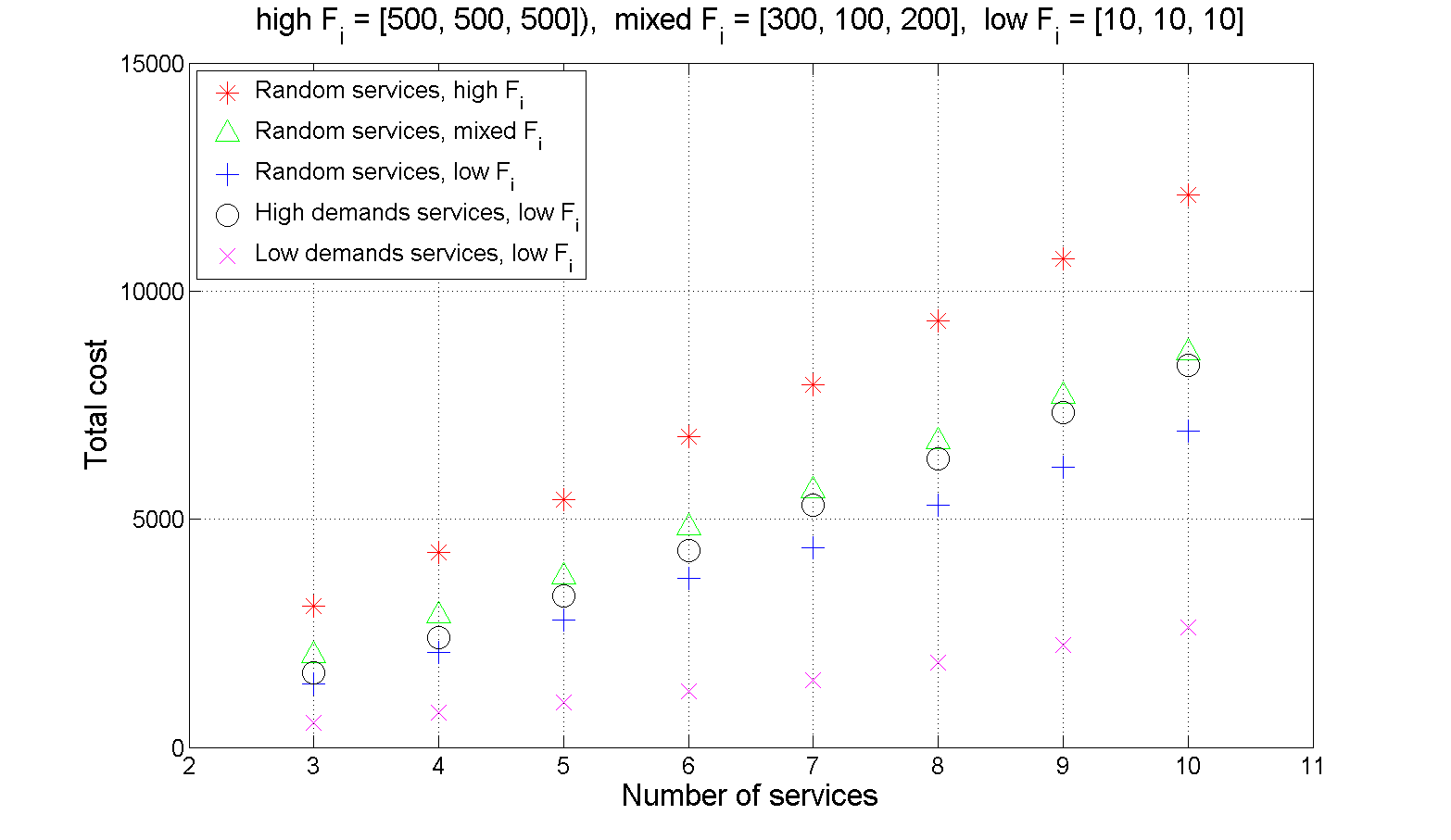}
\caption{Total cost vs number of services.}
\label{fig:cost}
\end{figure}

\begin{figure}
\centering
\includegraphics[width=1\columnwidth]{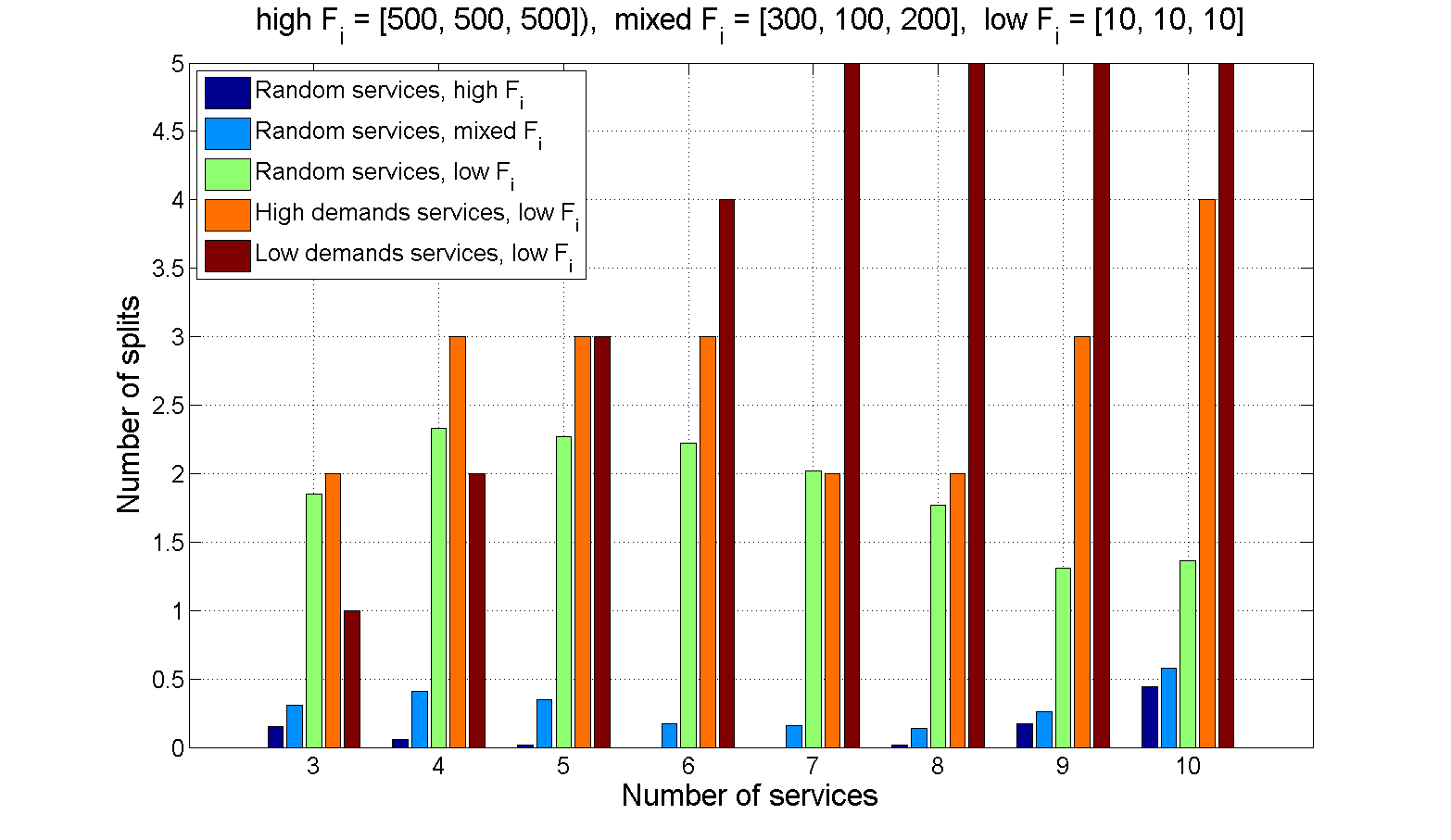}
\caption{Number of splits vs number of services.}
\label{fig:splits}
\end{figure}

\section{SIA complexity characterization}

\begin{theorem} The SIA is NP-Complete. \end{theorem}

\begin{proof}
The Partitioning Problem (PP) amounts to determining if a set of $\mathnormal{j}$ integers, of sum $\mathcal{S}$ can be partitioned into two subsets, each having a sum of  $\mathcal{S}/2$. The PP is a well-known NP-Complete problem. So, we base the proof on the construction of an instance of the problem from any instance of the Partitioning Problem, as follows. 

Assume that we have only one resource type on the interfaces available $(K=1)$. Let each element in the set of the PP be a service of the SIA problem instance and the value of each element be the resource demand $d_j$ of the corresponding $j$-th service. Additionally, let there be just two interfaces $ (I=2) $, each with resource availability $ b_i=\mathcal{S}/2, \forall i\in \{1,2\} $. 
We set  the overhead coefficients to zero $\alpha_{ij}=0, \forall i  \in \{1, 2\}, \forall j \in \{ 1,\dotsc,J\} $ and the resource activation cost to zero likewise $ c_{ij} = 0 $, while we fix the interface activation cost to one $ F_{i} = 1 $.

The constructed SIA instance is feasible, because (i) by construction the total resource availability on the two interfaces suffices to serve the aggregate demand and (ii) splitting service demand on more than one interfaces is allowed.

Consider a solution of the constructed SIA instance where no splitting occurs. If such a solution exists, then each service is assigned to one interface and by equation (\ref{fig:cost}) the cost will be equal to $J$. Furthermore, since any service split in two interfaces gives a cost of two, if splits exist in a solution, the cost will be at least $J+1$. Hence, by the construction of the instance, it becomes obvious that no value lower  than $J$ can be achieved. The recognition version of the SIA instance is to answer whether or not there is a solution for which the cost is at most some value, in our case $J$.

In any solution of SIA the service demand assigned on each interface will be $\mathcal{S}/2$.  If there is no split of services in a solution, then their assignment to the two interfaces is a partitioning of the integers in the PP with equal sum. So, if the answer to the original instance of the PP is yes, then by assigning to the two interfaces the services mapping to the elements of the solution subsets, no split will exist and the cost will be $J$. Thus, the answer to the recognition version of the SIA instance is yes. Conversely, if the answer to the SIA is yes, then there cannot be a split service so the assignment of services to the two interfaces is a valid PP solution. Therefore, solving the constructed SIA instance is equivalent to solving an arbitrary PP instance.

\end{proof}

\section{Simulation setup and results}

We performed several sets of simulations in Matlab to assess the cost and the number of splits for instances of three to ten services using different configuration of interfaces' costs and capacities as well as services' demands. 

Note that in this early presentation of our work, interfaces' capacities were adequate to provide an optimal solution for each optimization problem. When this doesn't hold true, this problem will entail jobs' scheduling issues. Interfaces' utilization costs ($c_{ik}$'s) were constant throughout the experiments and chosen such that they are not uniform among interfaces. However, the activation costs of the interfaces ($F_i$'s) were tuned in order to reflect the effect they may have splitting the services among several interfaces. 

Services' demands were chosen from three different classes to model the fact that they are not considerably arbitrary. Three sets of simulation setups were considered with the intention of modeling different sets that may arise in practice. The first two sets consisted only of low and high requirements demands respectively, whilst the third one was comprised of a mixture of demands requirements, which were chosen randomly. In this case, we ran the experiments 1000 times and averaged to evaluate the quantities of interest. Additionally, we tried three different activation costs for the random demands set.

\subsection{Total Cost}
The plots reflect the fact that the higher the activation cost ($F_i$), the more expensive it is to split. The optimal costs for different set of services are depicted in \figurename~\ref{fig:cost}.

If the activation cost is much higher than the cost of the interfaces, then the optimal total cost is higher in the case of services of random requirements than in the case of high demands services (with low activation cost).

Having mixed activation costs yields an optimal cost close to that of high demands services, because the latter cause more splits (see \figurename~\ref{fig:splits}) while the former has at most one job split on average.

A general remark on cost is that a less gradual increase in the optimal cost appears, when less splits will happen while the services increase simultaneously. For example, this is the case six or more random demands services are assigned to interfaces with low activation cost.

\subsection{Number of splits}

When the activation cost ($F_i$) is high in comparison to the utilization cost of the interfaces, less splits of services among interfaces happen. For instance, in \figurename~\ref{fig:splits} services of random requirements with high activation cost split at most once on average. On the other hand, when the activation cost is low (compared to the utilization cost) more splits occur.

Low requirements services split more than any other case to exploit the inexpensive (with regard to activation cost)  interfaces. More splits happen as the number of such services increase.

However, this is not the case regarding high demands services. More than six of such services causes a split less. The increase of splits for the case of nine and ten services can be attributed to the chosen interfaces' capacities - a split more was necessary to accommodate the demands of the extra served jobs.

\bibliographystyle{ieeetr}
\bibliography{arxiv_01}

\end{document}